\documentclass[runningheads]{llncs}
\usepackage{amsmath}
\usepackage{amssymb}
\usepackage{amsfonts}
\usepackage{tikz}
\usetikzlibrary{arrows}
\usetikzlibrary{decorations.markings}
\usetikzlibrary{shapes}

\tikzset{->-/.style={decoration={
  markings,
  mark=at position .5 with {\arrow{>}}},postaction={decorate}}}
\tikzset{->--/.style={decoration={
  markings,
  mark=at position .3 with {\arrow{>}}},postaction={decorate}}}
\tikzset{-->-/.style={decoration={
  markings,
  mark=at position .7 with {\arrow{>}}},postaction={decorate}}}
  
\newcommand{\s}{\mathcal{S}}
\newcommand{\N}{\mathcal{N}}
\newcommand{\G}{\mathcal{G}}

\newcommand{\C}{\mathrm{C}}
\newcommand{\Dom}{\mathrm{D}}
\newcommand{\Ran}{\mathrm{R}}

\newcommand{\tbarr}[2]{\scriptsize{\begin{array}{cc} #1 \\ #2 \end{array}}}
\newcommand{\tbarrr}[3]{\scriptsize{\begin{array}{cc} #1 \\ #2 \\ #3 \end{array}}}
\begin{document}
\title{Domain Range Semigroups and Finite Representations}
%
%
\author{Ja{\v s} {\v S}emrl \orcidID{0000-0001-7440-8867} \thanks{The author thanks Professor Robin Hirsch for supervision and insightful conversations about the work presented}}
\authorrunning{J. {\v S}emrl }
%
\institute{University College London, Gower St, London WC1E 6BT, UK
\email{j.semrl@cs.ucl.ac.uk}}
\maketitle              
\begin{abstract}
Relational semigroups with domain and range are a useful tool for modelling nondeterministic programs. We prove that the representation class of domain-range semigroups with demonic composition is not finitely axiomatisable. We extend the result for ordered domain algebras and show that any relation algebra reduct signature containing domain, range, converse, and composition, but no negation, meet, nor join has the finite representation property. That is any finite representable structure of such a signature is representable over a finite base. We survey the results in the area of the finite representation property.

\end{abstract}
\section{Introduction}
Formal reasoning about programs and their correctness is an important, yet a demonstrably difficult task and many well known approaches have been proposed. Algebraically speaking, a deterministic program is a partial function mapping from the state space to itself. Generalising this, to account for nondeterminism, we can say that a program (deterministic or nondeterministic) is a binary relation over the state space. This ability to naturally express such concepts motivates the endeavour of formalising the logic of binary relations.

A formalisation of this sort is found in Relation Algebra, obtained by extending the language of Boolean Algebra with operations specific to binary relations. This enables us to reason about the behaviour of binary relations in an abstract manner. However, these algebras are also very badly behaved, with an abundance of undecidability results, see \cite[Part V]{hirsch2002relation}. A possible way of combating this is by dropping some operations from the language, sacrificing the ability to encapsulate the behaviour of relational calculus in exchange for decidability of certain decision problems. We will formally define some of these and how to prove positive properties later in this section.

Here we examine some of these favourable properties, or lack thereof, for languages containing domain and range, and put them in the bigger context of relation algebra reduct languages. We chose this subset of languages as they were found useful in algebraically reasoning about correctness of nondeterministic programs, see Section~\ref{sec:relWork} for more details.

But first, some definitions. Let $X$ be a base set. \emph{Domain} ($\Dom$) and \emph{range} ($\Ran$) are operations, defined for some relation $R \subseteq X \times X$ as
$$\Dom(R) = \{(x,x) \mid \exists y: (x,y) \in R\} \hspace{0.5cm} \Ran(R) = \{(y,y) \mid \exists x: (x,y) \in R\}$$
and together with composition, they form the signature of domain-range semigroups. However, relational composition is not always interpreted in the same way. Two examples of interpretations include the \emph{angelic} or ordinary composition (denoted $;$) and \emph{demonic} composition (denoted $*$), defined below for $R,S \subseteq X \times X$
\begin{gather*}
    R;S = \{(x,z) \mid \exists y: (x,y) \in R \wedge (y,z) \in S\} \\
    R*S = \{(x,y) \in R;S \mid \forall z: (x,z) \in R \Rightarrow (z,z) \in \Dom(S)\}
\end{gather*}
Whilst the first definition seems pretty intuitive, the second one may appear a bit odd, even arbitrary, so let us have a closer look. The operation is motivated in the behaviour of a nondeterministic machine when the demon is in control of nondeterminism. Imagine the relations $R,S$ were programs over the state space $X$. The pair $(x,y) \in R;S$ is included in $R*S$ if and only if there is no run from $R$ to some $z$ from which $S$ aborts or loops forever, i.e. $(z,z) \notin \Dom(S)$. Should such a run exist, the demon will take the opportunity and abort the computation. For more details on this refer to \cite{hirsch2021demonic}.

Any $\{D,R,;\}$- or $\{D,R,*\}$-structure $\s$ with an underlying set $S \subseteq \wp(X \times X)$ for some base $X$ and operations interpreted relationally (as defined above) is \emph{proper}. Let $\tau$ be a signature of operations that are well defined for binary relations. The \emph{representation class} for $\tau$, denoted $R(\tau)$, is the class of all proper $\tau$-structures, closed under isomorphic copies. An isomorphism $\theta$ that maps a representable structure to a proper structure is called a \emph{representation}.

A representation is finite if the base set $X$ of the proper image is finite. If all finite members of $R(\tau)$ have finite representations, we say that the signature has the \emph{finite representation property} (FRP).

The two properties described above are of special interest to us. This is because they both guarantee the decidability of determining membership in $R(\tau)$ for finite structures, also known as the \emph{representability decision problem}. Although the properties both ensure decidability of the said decision problem, they in no way follow from each other. This provides us with two non trivial questions for each Relation Algebra reduct language that, given either is answered affirmatively, provide us with a decidability guarantee.

Here, we answer \cite[Question 4.9]{jackson2019domain} and show that $R(D,R,*)$ is not finitely axiomatisable. We do so by defining a two-player game that corresponds to a recursively enumerable axiomatisation of the representation class. Then we show that for each finite subset of this axiomatisation has a non-representable model. By compactness of first order logic, we are able to reach a contradiction under the assumption of finite axiomatisability.

Then we show that any relation algebra reduct signature containing domain, range, converse and composition, but no negation, meet, nor join has the finite representation property. This is an extension of a previous finite representation property result for ordered domain algebras \cite{hirsch2013meet}. We conclude by putting the result in a larger context of finite representation property for all reduct signatures of relation algebra. We survey the existing results and raise some open questions in the area.

\section{Motivation and Context}
\label{sec:relWork}
In this section we take a closer look at the related work and motivate the problems. We have seen that structures of relations provide us with a natural way of formally reasoning about nondeterministic programs \cite{dijkstra2012predicate}. In \cite{desharnais2009domain}, a good intuition on how to use structures with domain and range to model program control flow using semigroups with domain and range -- functional for deterministic, and relational for nondeterministic programs. This allows us to express partial correctness equationally.

However, to extend this to total correctness, we have to turn to the demon. Demonic calculus was introduced to model the behaviour of programs, should the demon be in control of making nondeterministic decisions. Recently, it has been shown we may take this to our advantage and introduce equations to model total correctness. One such approach expresses total correctness using the domain and demonic composition \cite{hirsch2020algebra} and another using ordinary composition and the bottom element of the demonic lattice \cite{hirsch2021demonic}.

These applications motivate our search for computational guarantees. As we have discussed, this includes looking for finite axiomatisability of the representation class and the finite representation property. A major negative result is shown with $R(\Dom,\Ran,;)$ and $R(\Dom, ;)$ having no finite axiomatisation \cite{hirsch2011axiomatizability}.

Both of these two signatures have the finite representation property open. However, one may add the partial ordering, converse, the identity, and the empty relation to obtain the signature of ordered domain algebras. Surprisingly, this signature has both the finite representation property, as well as a finitely axiomatisable representation class \cite{hirsch2013meet}. Another interesting result is the axiomatisation of $R(\Dom, *)$ is not only finite, but also the same as that of representable domain semigroups of partial functions \cite{hirsch2021axioms}. Furthermore, the equational theories of both $R(\Dom, \Ran, ;)$ and $R(\Dom, \Ran, *)$ are finitely axiomatisable \cite{jackson2019domain}.

Finally, it is important to note that although the finite axiomatisability of the representation class and the finite representation property both guarantee the decidability of the representation decision problem, neither is stronger or weaker than the other. We have seen an example of a signature with both properties in the ordered domain algebras, as well as the full signature of relation algebras with neither property. However, you can find signatures with finitely axiomatisable representation class but no FRP, like meet-lattice semigroups \cite{bredihin1978representations,neuzerling2016undecidability}, and semigroups with demonic refinement \cite{hirsch2020finite} with FRP, but non finitely axiomatisable representation class.

\section{Networks and Representation by Games}
In this section we outline a representation game that will help us prove the non finite axiomatisability result of $R(\Dom,\Ran,*)$. This argument is based on \cite{hirsch2002relation}, but defined for this specific signature. The proofs presented are outlines, however, they are more detailed in parts where it is necessary to show the argument can be feasibly used to show results for demonic composition. For full details of proofs, see \cite[Chapter 7]{hirsch2002relation}.

On an intuitive level, this approach entails defining a game where a player is challenged to build a representation, on a step by step basis over a predetermined number of moves. The design of our game must be such that the player challenged will have a winning strategy if and only if they can survive the game of any length.

We then, for every natural number, define a formula that corresponds to a winning strategy for a game of that length. This means that we have defined a recursively enumerable theory that axiomatises the representation class.

In later sections we define, for each length of the game, an unrepresentable structure where the player challenged has got the winning strategy. This will enable us to use the compactness of first order logic to reach a contradiction under the assumption of finite axiomatisability.

Now, we will define these concepts more formally. A network $\mathcal{N} = (N, \bot, \top)$ where $\bot, \top: N \times N \rightarrow \wp(\mathcal{S})$ and $\s$ is some $\{\Dom, \Ran, *\}$-structure. We say it is \emph{consistent} if and only if
\begin{gather*}
    \forall x,y \in N: \top(x,y) \cap \bot(x,y) = \emptyset \\
    \forall x, y \in N, \forall s,t \in \s: \Big( s \in \top(x,y) \wedge \big(s = \Dom(t) \vee s = \Ran(t)\big)\Big) \Rightarrow x=y 
\end{gather*}

Now, let us define for any $a,b \in \s$ the two networks $\N_{ref}[a,b]$ and $\N_{nref}[a,b]$ as follows
\begin{gather*}
    \N_{ref}[a,b] = (\{x\}, \{(x,x) \mapsto \{b\}\}, \{(x, x) \mapsto \{a\}\})\\
    \N_{nref}[a,b] = (\{x,y\}, \{(x,y) \mapsto \{b\}\}, \{(x,y) \mapsto \{a\}\})
\end{gather*}
And all other pairs map to $\emptyset$ for $\top,\bot$.

We also define two operations $+_\top[\N,x,y,a], +_\bot[\N,x,y,a]$ which take a network $\N = (N,\bot, \top)$, some $x,y \in N \dot\cup \{x_+\}$ and some $a \in \s$ and return 
\begin{gather*}
    +_\top[\N,x,y,a] = (N \cup \{x,y\}, \bot, \top^+)\\
    +_\bot[\N,x,y,a] = (N \cup \{x,y\}, \bot^+, \top)
\end{gather*}
where $\top^+(v,w)$ is the same as $\top(v,w)$, or $\emptyset$ (if $\top(v,w)$ is undefined), for all $v,w$, except for $x,y$ where $a$ is also added to $\top^+(x,y)$. Similarly for $\bot^+$.

A network $\N' = (N',\top', \bot')$ is said to \emph{extend} $\N = (N, \top , \bot)$, denoted $\N \subseteq \N'$ if and only $N \subseteq N'$ and for all $x,y \in N$ we have $\top(x,y) \subseteq \top'(x,y), \bot(x,y) \subseteq \bot'(x,y)$. Clearly, both $+_\top, +_\bot$ for $\N$ with any operands are extensions of $\N$. Furthermore, observe how inconsistency is inherited under extensions.

We can now define a game for a $\{D,R,*\}$-structure $\mathcal{S}$. It is played by two players $\forall, \exists$, we will call them Abelard and Eloise. The game, denoted $\Gamma_n(\mathcal{S})$, starts with the initialisation (zeroth) move and then continues for $n$ moves where $0 < n \leq \omega$. Let $k \leq n$. At $k$th move $\forall$ challenges $\exists$ to return a $\N_k$ such that $\N_0 \subseteq \N_1 \subseteq ... \subseteq \N_n$. $\forall$ wins the game if and only if $\exists$ introduces an inconsistent network.

\paragraph{Initialisation.} $\forall$ picks a pair $a \neq b \in \s$ and $\exists$ returns $\N_0$ that is an extension of $\N_{ref}[a,b]$, $\N_{nref}[a,b]$, $\N_{ref}[b,a]$ or $\N_{nref}[b,a]$.

\paragraph{Witness Move.} $\forall$ picks a pair of nodes $x,z$ in the network $\N_k$ and a pair of elements $a,b \in \s$ such that $a * b \in \top(x,z)$. $\exists$ picks a $y \in N \dot\cup \{x_+\}$ and returns $\N_{k+1} \supseteq +_\top[+_\top[\N,x,y,a],y,z,b]$, see Figure~\ref{fig:witcd} left.
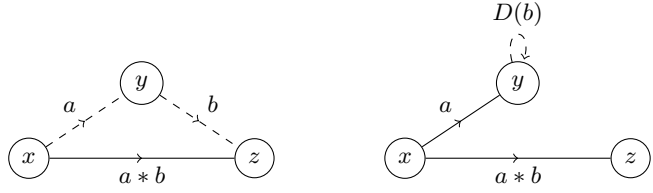
\begin{figure}[t]
    \centering
    \begin{tikzpicture}
        \node[draw,circle] (x) at (0,0) {$x$};
        \node[draw,circle] (z) at (3,0) {$z$};
        \node[draw,circle] (y) at (1.5,1) {$y$};
        
        \path (x) edge[->-] node[below]{$a*b$} (z);
        \path (x) edge[->-, dashed] node[above left]{$a$} (y);
        \path (y) edge[->-, dashed] node[above right]{$b$} (z);

        \node[draw,circle] (x) at (5,0) {$x$};
        \node[draw,circle] (z) at (8,0) {$z$};
        \node[draw,circle] (y) at (6.5,1) {$y$};
        
        \path (x) edge[->-] node[below]{$a*b$} (z);
        \path (x) edge[->-] node[above left]{$a$} (y);
        \path (y) edge[loop above, dashed] node[above]{$D(b)$} (y);

    \end{tikzpicture}
    \caption{Witness Move (left) and Composition-Domain Move (right)}
    \label{fig:witcd}
\end{figure}
\paragraph{Composition-Domain Move.} $\forall$ picks, some $x,y,z$ with $a \in \top(x,y)$ and $a * b \in \top(x,z)$ and $\exists$ must return $\N_{k+1} \supseteq +_\top[\N, y, y, \Dom(b)]$, see Figure~\ref{fig:witcd} right.

\paragraph{Composition Move.} $\forall$ picks some $x,y,z \in \N_k$ along with $a,b$ such that $a \in \top(x,y)$ and $b \in \top(y,z)$. $\exists$ has a choice between returning $\N_{k+1} \supseteq +_\top[\N, x, z, a * b]$ (Figure~\ref{fig:comp} left) and $\N_{k+1} \supseteq +_\bot[+_\top[\N,x,w,a],w,w,\Dom(b)]$ where she picks a $w \in N \dot\cup \{x_+\}$ (Figure~\ref{fig:comp} right).
\begin{figure}[t]
    \centering
    \begin{tikzpicture}
        \node[draw,circle] (x) at (0,0) {$x$};
        \node[draw,circle] (z) at (3,0) {$z$};
        \node[draw,circle] (y) at (1.5,1) {$y$};
        
        \path (x) edge[->-, dashed] node[below]{$a*b$} (z);
        \path (x) edge[->-] node[above left]{$a$} (y);
        \path (y) edge[->-] node[above right]{$b$} (z);

        \node[draw,circle] (x) at (5,0) {$x$};
        \node[draw,circle] (z) at (8,0) {$z$};
        \node[draw,circle] (y) at (6.5,0) {$y$};
        \node[draw,circle] (w) at (6.5,1) {$w$};
        
        \path (y) edge[->-] node[below]{$b$} (z);
        \path (x) edge[->-] node[below]{$a$} (y);
        \path (w) edge[loop above, dashed] node[above]{$\neg D(b)$} (w);
        \path (x) edge[->-, dashed] node[above left]{$a$} (w);

    \end{tikzpicture}
    \caption{Composition Move}
    \label{fig:comp}
\end{figure}
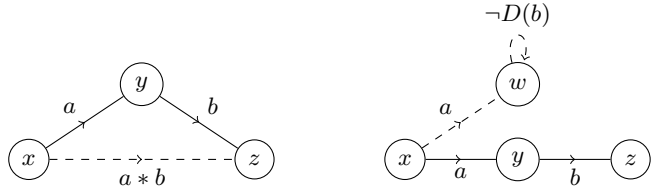
\paragraph{Domain-Range Move.} $\forall$ picks $x,y \in N_n, a \in \s$ such that $a \in \top(x,y)$ and $\exists$ must return $\N_{k+1} \supseteq +_\top[+_\top[\N,x,x,\Dom(a)],y,y, \Ran(a)]$.

\paragraph{Domain Move.} $\forall$ picks a node $x$ and an $a \in \s$ such that $\Dom(a) \in \top(x,x)$ and $\exists$ must pick a node $y \in N \dot\cup \{x_+\}$ and return $\N_{k+1} \supseteq +_\top[\N,x,y,a]$

\paragraph{Range Move.} $\forall$ picks a node $y$ and an $a \in \s$ such that $\Ran(a) \in \top(y,y)$ and $\exists$ must pick a node $x \in N \dot\cup \{x_+\}$ and return $\N_{k+1} \supseteq +_\top[\N,x,y,a]$

\begin{lemma}
\label{lem:winStrat}
A countable $\{\Dom, \Ran, *\}$-structure is representable if and only if $\exists$ has a winning strategy for $\Gamma_\omega(S)$.
\end{lemma}

\begin{proof}
If the structure is representable, $\exists$ can play the game by mapping the responses from the representation. Conversely, if $\exists$ has a winning strategy for $\Gamma_\omega(\s)$, she must also have the winning strategy for any length of the game where $\forall$ schedules moves in the way that eventually every move will be called and the $\top$ label of the network will in the limit be closed under composition, domain-range moves and saturated under witness, domain and range moves. Since the structure is countable, $\forall$ can schedule moves in this manner. Take the limit network, call it $\N_\omega[a \neq b]$, after such a play with the initialisation pair $a \neq b$. Observe how due to saturation and closure, the $\top$ outlines a mapping from $\s$ to $N \times N$ that represents $\Dom, \Ran, *$ correctly and ensures that $a,b$ map to different relations. Thus a disjoint union $\dot\bigcup_{a \neq b} \N_\omega[a \neq b]$ is a representation of $\s$.
\qed\end{proof}

\begin{lemma}
\label{lem:sigma}
For every $n<\omega$, there exists a first order formula $\sigma_n$ such that $\exists$ has a winning strategy for $\Gamma_n(\s)$ if and only if $\s \models \sigma_n$. Furthermore, the first order theory $\Sigma = \{\sigma_i \mid i < \omega \}$ axiomatises $R(\Dom, \Ran, *)$.
\end{lemma}

\begin{proof}
Let us define a \emph{variable network} in a slightly different manner with the mappings $\top, \bot: N \times N \rightarrow \wp(\mathrm{Vars})$. A valuation $v: \mathrm{Vars} \rightarrow \s$ defines a conventional network $v(\N)$. This allows us to define a formula $\phi_n(\N)$ in a way that, together with a valuation $v: \s \rightarrow \mathrm{Vars}$, $\exists$ can survive the \emph{conservative} play of the game for $n$ more moves, starting from $v(\N)$. By conservative, we mean that $\exists$ plays the network requested without proper extensions.

In the base case, observe how $v(\N)$ only needs to be consistent and thus
$$
    \phi_0(\N) = \bigwedge_{\tbarr{x \neq y \in N}{s \in \top(x,y)}} \neg\exists t: s = \Dom(t) \vee s = \Ran(t) \;\;\wedge \bigwedge_{\tbarrr{x, y \in N}{s \in \top(x,y)}{t \in \bot(x,y)}} s \neq t
$$
In the induction case, if $\phi_n[\N]$ signifies that $\exists$ can survive for $n$ more moves, simply define $\phi_{n+1}$ as
{\allowdisplaybreaks \begin{align*}
    \phi_{n+1}(\N) = & \bigwedge_{\tbarr{x,z \in N}{s \in \top(x,z)}} \!\forall t,u: s = t * u \Rightarrow \bigvee_{y \in N \dot\cup \{x_+\}} \phi_n(+_\top[+_\top[\N, x,y,a],y,z,b])\\
    &\wedge \bigwedge_{\tbarr{x,y,z \in N}{t \in \top(x,y), u \in \top(y,z)}} \forall s: s = t * u \rightarrow \Bigg( \phi_n(+_\top[\N, x,z, s] \\
   &  \hspace{0.5cm}\vee\;  \forall v: v = \Dom(t) \Rightarrow \bigvee_{w \in N\dot\cup\{x_+\}}\phi_n(+_\bot[+_\top[\N,x,w,t],w,w,v]) \Bigg)\\
   &\wedge \bigwedge_{\tbarrr{x,y,z \in N}{t \in \top(x,y)}{s \in \top(x,z)}} \forall u,v: (s = t * u \wedge v = \Dom(u)) \Rightarrow \phi_n(+_\top[\N,y,y,v])\\
   &\wedge \bigwedge_{\tbarr{x,y \in N}{s \in \top(x,y)}} \forall t,u: \Big(t = \Dom(s) \wedge u = \Ran(s)) \Rightarrow\\
   & \hspace{3.5cm}\phi_n(+_\top [+_\top [\N,x,x,t],y,y,u]\Big)\\
   & \wedge \bigwedge_{\tbarr{x \in N}{s \in \top(x,x)}} \forall t: \Dom(t) = s \Rightarrow \bigvee_{y \in N\dot\cup\{x_+\}}\phi_n(+_\top[\N,x,y,t])\\
   & \wedge \bigwedge_{\tbarr{y \in N}{s \in \top(y,y)}} \forall t: \Ran(t) = s \Rightarrow \bigvee_{x \in N\dot\cup\{x_+\}} \phi_n(+_\top[\N,x,y,t])
\end{align*}}
Thus $\exists$ can win a conservative game $\Gamma_n(\s)$ if and only if $\s \models \sigma_n$ where
\begin{align*}
\sigma_n = \forall s,t: s \neq t \Rightarrow \bigg(&\phi_n(\N_{ref}[a,b]) \vee \phi_n(\N_{nref}[a,b])\\ & \vee \phi_n(\N_{ref}[b,a]) \vee \phi_n(\N_{nref}[b,a])\bigg)
\end{align*}
Since inconsistencies in networks are inherited in extensions, it is true that for countable structures if $\exists$ has a winning strategy for conservative plays of $\Gamma_n(\s)$, she will also have a winning strategy for any play of $\Gamma_n(\s)$. Furthermore, as inconsistency is inherited in extensions, if $\s \models \Sigma$, $\exists$ has a winning strategy for $\Gamma_\omega(\s)$. Thus for all countable $\s$, $\s \in R(\Dom, \Ran, *)$ if and only if $\s \models \Sigma$. As the representation class is pseudoelementary, it is closed under elementary equivalence, and by L{\"o}wenheim-Skolem Theorem, we conclude $\s \models \Sigma$ is both sufficient and necessary for membership, even for uncountable structures.  
\qed\end{proof}

\section{Demonic Refinement}
\label{sec:demrf}
Before we move on to defining structures used to prove non finite axiomatisability, we will quickly have a look at the demonic lattice. We discuss in Section~\ref{sec:relWork} that the demonic lattice has found use in algebraically modelling total correctness. However, in this section, it will help us show that the structures we will use in the argument are in fact non-representable.

We do so by defining demonic refinement, the partial ordering predicate arising from the demonic lattice. Furthermore, we observe that even though the predicate is not in the signature, some pairs of elements of a representable $\{\Dom,\Ran,*\}$-structure will always be represented as demonic refinement pairs.

Now assume that a $\{\Dom, \Ran, *\}$-structure has a cycle of elements where each element is a demonic refinement of its successor. As the predicate is a partial order, it means by antisymmetry and transitivity that these distinct elements will map the same binary relation in any representation and thus the structure is not representable.

Now let us define demonic refinement for $R,S \subseteq X \times X$ as
$$R \sqsubseteq S \Longleftrightarrow (\Dom(S) \subseteq \Dom(R) \wedge \Dom(S);R \subseteq S)$$

This is motivated, again, with the demon in control of nondeterminism. Imagine $R,S$ were programs over the state space $X$. If the demon is given the choice to run $R$ or $S$, he will always run $S$. This is because when we are outside the domain of $S$, running $S$ rather than $R$ will result abort and when in the domain of $S$ it will maximise the odds of reaching an erroneous state.

Now we recursively define a predicate $\preceq$ using infinitary $\{\Dom,\Ran,*\}$-formula such that for every structure $\s$ with a representation $\theta$ we will have $\forall s,t \in \s: s \preceq t \Rightarrow s^\theta \sqsubseteq t^\theta$. We take advantage of the fact that sometimes non domain elements may compose to a domain element, and define $\preceq_1$. Then we inductively close the predicate under monotonicity and transitivity. More formally, we say that
$$s \preceq_1 t \Longleftrightarrow \exists u,v: \Dom(u * v) = u * v \wedge s = \Ran(u * \Dom(v)) \wedge t = s * v * u$$
$$s \preceq_{n+1} t \Longleftrightarrow \left(\begin{array}{rl}
     & \exists s',t',u,v: s' \preceq_n t' \wedge s = u * s' * v \wedge t = u * t' * v\\
     \vee & \exists v: s \preceq_n v \wedge v \preceq_{n} t  
\end{array}\right)$$
and $\preceq = \bigcup_{n<\omega} \preceq_n$

\begin{lemma}
For any $s,t \in \s$, if $s \preceq t$, it is true that for any representation $\theta$ we have $s^\theta \sqsubseteq t^\theta$. \end{lemma}

\begin{proof}
We show this by induction over $n$.

In the base case, we see that there exists a $u,v$ such that $u * v = \Dom(u * v)$ and $s = \Ran(u * \Dom(v))$ and $t = s * v * u$. First see how if $(x,x) \in t^\theta$, there must exist a witness for $s * v * u$ and since $s$ is a range element, it must hold that $(x,x) \in s^\theta$. Since $\Dom(s^\theta) = s^\theta$, we have
$\Dom(t^\theta) \subseteq \Dom(s^\theta)$. Furthermore, assume that $(x,x) \in \Dom(t^\theta)$ and $(x,x) \in s^\theta$. See how there must exist a $y$ such that $(y,x) \in (u * \Dom(v))^\theta$. There must also exist a $z$ such that $(x,z) \in v^\theta$. Since $(y,y) \in \Dom(u * \Dom(v))^\theta$, we can see that $(y,z) \in (u * v)^\theta$ and since $u * v$ is a domain element, $y = z$. And because $(x,x) \in \Dom(t)^\theta$ and because $(x,z) \in (s * v)^\theta$ and $(z,x) \in u^\theta$, we conclude $(x,x) \in (s * v * u)^\theta = t^\theta$.

The induction case follows from the fact that $\sqsubseteq$ is transitive as well as left and right monotone for $*$ as discussed in \cite{hirsch2020finite}.
\qed\end{proof}
The use of refinement cycles may seem similar to \cite{hirsch2011axiomatizability} where the predicate $\triangleleft$ is defined as the monotone, transitive closure of $\Dom(s) ; \Dom(t) \triangleleft \Dom(t)$ to signify ordinary inclusion ($\leq$) for the angelic signature. However, for the demonic signature, $\triangleleft$ can be simply described as $\Dom(s) * t \triangleleft t$ as the following axiom is sound
\begin{gather*}
\forall s,t: \Dom(s * \Dom(t)) * s = s * \Dom(t)\end{gather*}
Thus, $\triangleleft$ does not show useful when trying to show $R(\Dom, \Ran, *)$ is not finitely axiomatisable, as avoiding cycles of $\triangleleft$ can be described in a single axiom. 

\section{$R(\Dom,\Ran,*)$ is Not Finitely Axiomatisable}
\label{sec:nfa}
We can now define the non representable structures for every $n<\omega$ for which $\exists$ will have a winning strategy in $\Gamma_n$. First we use the demonic refinement predicate, defined in Section \ref{sec:demrf}, to show these are not representable as they include a refinement cycle. Then we show by induction that $\exists$ will have a winning strategy for $n$ moves in the representation game. Using the compactness trick, we show that the representation class is not finitely axiomatisable.

For every $n < \omega$, let $N=2n+1$. Define a $\{\Dom, \Ran, *\}$-structure $\s_n$, with the following underlying set
$$\{0,d,r\} \cup \{m_i,\varepsilon_i, a_i,b_i,c_i,d_i,ac_i,acd_i,cdb_i, db_i, ab_i \mid 0 \leq i < N \}$$
$0,d,r,m_i, \varepsilon_i$ are the domain-range elements, idempotent with respect to composition, and disjoint, i.e. composition of two distinct domain-range elements evaluates to 0. We now examine domain-range elements, see visualisation in Figure~\ref{fig:sn}. For all $i < N$, we have
\begin{gather*}
    d = \Dom(a_i) = \Dom(ac_i) = \Dom(acd_i) = \Dom(ab_i) \\
    m_i = \Dom(c_i) = \Dom(b_i) = \Dom(cdb_i) = \Ran(a_i) = \Ran(d_i) = \Ran(cd_i)\\
    \varepsilon_i = \Dom(d_i) = \Dom(db_i) = \Ran(c_i) = \Ran(ac_i)\\
    r = \Ran(ab_i) = \Ran(cdb_i) = \Ran(db_i) = \Ran(b_i)
\end{gather*}
The reader may find it helpful to pay close attention to Figure~\ref{fig:sn} while we define the compositions. First, we say that
$$ d_i* c_i = \varepsilon_i \hspace{1cm} c_i * d_i = cd_i \hspace{1cm} cd_i * cd_i = cd_i $$
for every $i<N$. Furthermore, some elements will result in a composition with an index increasing by one, namely
$$a_i * cdb_i = ab_{i+1} \hspace{0.5cm} acd_i * cdb_i = ab_{i+1} \hspace{0.5cm} ac_i * db_i = ab_{i+1} \hspace{0.5cm}  acd_i * b_i = ab_{i+1} $$
for $i<N$ where $+$ denotes addition modulo $N$. Composition results below are defined more naturally
$$
\begin{array}{ccccc}
    cd_i * c_i = c_i  \hspace{0.2cm}&\hspace{0.2cm} d_i * cd_i = d_i \hspace{0.2cm}&\hspace{0.2cm}  a_i * b_i = ab_i \hspace{0.2cm}&\hspace{0.2cm}  a_i * c_i = ac_i \hspace{0.2cm}&\hspace{0.2cm}  a_i * cd_i = acd_i \\
    c_i * db_i = cdb_i & d_i* b_i = db_i & ac_i * d_i = acd_i & acd_i * c_i  = ac_i & cd_i * b_i = cdb_i
\end{array}
$$
 All other compositions are either the mandatory domain-range compositions or they evaluate to $0$.
 
The following two Lemmas will now show that although $\s_n$ is not representable, $\exists$ will be able to maintain consistency in the network for $n$ moves.

\begin{figure}
    \centering
    \begin{tikzpicture}
    \coordinate [draw, circle](d1) at (0,0);
    \coordinate [draw, circle](r) at (8,0);
    \coordinate [draw, circle](ml) at (4,2);
    \coordinate [draw, circle](e) at (4,4);
    \node at(-0.3,4.7){$0 \leq i < N$};
    \draw[dashed](-1.2,0.6) rectangle (9,5);
    
    \path (d1)edge[loop left] node[above]{$d$}(d1);
    \path (r)edge[loop right] node[above]{$r$}(r);
    \path (ml)edge[loop below] node[below]{$m_i, cd_i$}(ml);
    \path (e)edge[loop above] node[above]{$\varepsilon_i$}(e);
    
    \path(d1)edge[->-]node[below]{$0 \leq j < n: ab_j$}(r);
    \path(d1)edge[->-]node[left, pos=0.7]{$a_i, acd_i\;\;$}(ml);
    \path(d1)edge[->-, bend left = 30]node[left]{$ac_i$}(e);
    \path(ml)edge[->-, bend left = 30]node[left]{$c_i$}(e);
    \path(e)edge[->-, bend left = 30]node[right]{$d_i$}(ml);
    \path(ml)edge[->-]node[right, pos=0.3]{$\;b_i, cdb_i$}(r);
    \path(e)edge[->-, bend left = 30]node[right]{$db_i$}(r);
    \end{tikzpicture}
    \caption{Visualisation of $\s_n$}
    \label{fig:sn}
\end{figure}
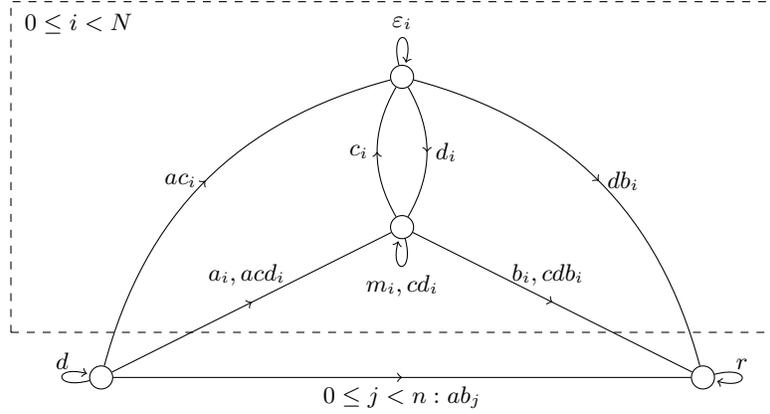

\begin{lemma}
\label{lem:nonrep}
$\s_n$ is not $\{\Dom, \Ran, *\}$-representable.
\end{lemma}

\begin{proof}
Observe how $m_i \preceq c_i * d_i$ and thus $ab_i = a_i * m_i * b_i \preceq a_i * c_i * d_i * b_i = ab_{i+1}$ for all $i < N$. This means by transitivity of $\preceq$ that for all $i,j < N$ we have $ab_i \preceq ab_j$. Now assume that there existed a representation $\theta$. We would have $ab_i^\theta \sqsubseteq ab_j^\theta, ab_j^\theta \sqsubseteq ab_i^\theta $, even where $i \neq j$. Since $\sqsubseteq$ is antisymmetric, we would have $ab_i^\theta = ab_j^\theta$ for $i \neq j$.  Therefore, no such $\theta$ can exist.
\qed\end{proof}

\begin{lemma}
\label{lem:ws}
For all $n < \omega$, $\exists$ has a winning strategy for $\Gamma_n(\s_n)$
\end{lemma}

\begin{proof}
First see how $\exists$ may play in a way that she returns a network that is closed under composition, composition domain and domain-range moves. For composition moves, she always chooses to add $a*b$ to the label, rather than adding a node with $D(b)$ in its $\bot$ label. Furthermore, she may set the $\top$ label in a way that for all $(x,y)$
$$\top_{k+1} \subseteq \top_{k}(x,y) \cup \{ ab_{i+1} \mid ab_i \in \top_k(x,y) \}$$
where $+$ is modulo $N$ and
$$\begin{array}{cc}
    a_i \in \top_k(x,y) \Rightarrow acd_i \in \top_k(x,y)\;\; &\;\; m_i \in \top_k(x,y) \Rightarrow cd_i \in \top_k(x,y)\\
    b_i \in \top_k(x,y) \Rightarrow cdb_i \in \top_k(x,y)\;\; &\;\; 0 \notin \top(x,y)
\end{array}$$
as well as ensure that domain-range elements are only added to reflexive edge $\top$ labels. If $m_i \in \top(x,x)$, there exists at most one $y$ such that $c_i \in \top(x,y) \vee d_i \in \top(y,x)$ and if $cd_i \in \top(x,x')$, the $y$ must be the same for $x, x'$. to prevent compositional closure from adding $m_i$ to a $\top(z,w), z \neq w$.

In the base case, observe how for every $s \neq t$, it is possible to find either $s$ or $t$ to put in $\top(x,y)$ of the initialisation network. Without loss, if $s = 0$ or $s=a_i, t=acd_i$ or $s=m_i, t = cd_i$ or $s=b_i, t=cdb_i$ she has to play $t$. Otherwise, she is free to play either $s$ or $t$, making sure that she plays the reflexive network if and only if she opts to play a domain-range element.

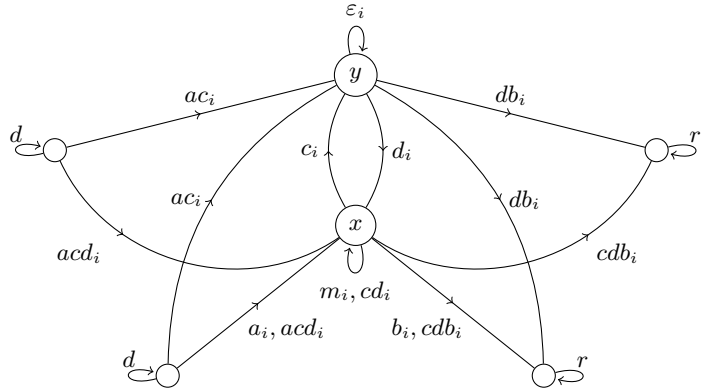
\begin{figure}
    \centering
        \begin{tikzpicture}
    \coordinate [draw, circle](d1) at (1.5,0);
    \coordinate [draw, circle](d3) at (0,3);
    \coordinate [draw, circle](r) at (6.5,0);
    \coordinate [draw, circle](r3) at (8,3);
    \node [draw, circle](ml) at (4,2) {$x$};
    \node [draw, circle](e) at (4,4) {$y$};
    
    \path (d1)edge[loop left] node[above]{$d$}(d1);
    \path (d3)edge[loop left] node[above]{$d$}(d3);
    \path (r)edge[loop right] node[above]{$r$}(r);
    \path (r3)edge[loop right] node[above]{$r$}(r3);
    \path (ml)edge[loop below] node[below]{$m_i, cd_i$}(ml);
    \path (e)edge[loop above] node[above]{$\varepsilon_i$}(e);
    
    \path(d1)edge[->-]node[right, pos=0.3]{$\;\;a_i, acd_i$}(ml);
    \path(d1)edge[->-, bend left = 30]node[left]{$ac_i$}(e);
    \path(d3)edge[->--, bend right=50]node[below left, pos=0.3]{$acd_i\;\;$}(ml);
    \path(d3)edge[->-]node[above]{$ac_i$}(e);
    \path(ml)edge[->-, bend left = 30]node[left]{$c_i$}(e);
    \path(e)edge[->-, bend left = 30]node[right]{$d_i$}(ml);
    \path(ml)edge[->-]node[left, pos=0.7]{$b_i, cdb_i\;\;$}(r);
    \path(e)edge[->-, bend left = 30]node[right]{$db_i$}(r);
    \path(ml)edge[-->-, bend right=50]node[below right, pos=0.7]{$cdb_i$}(r3);
    \path(e)edge[->-]node[above]{$db_i$}(r3);
    \end{tikzpicture}
    \caption{Compositions with $cd_i$}
    \label{fig:cd}
\end{figure}

In the induction case, as the network is closed under domain-range, composition and composition-domain moves, $\forall$'s only non-redundant move options are composition, domain, and range.

For the domain move, $\exists$ may add a new node unless $c_i$ is requested on some $x$. In that case, she must pick to add $c_i$ to $\top(x,y)$ to the designated $y$ if such $y$ exists, otherwise create such a $y$ and close it under all the necessary moves to maintain the induction hypothesis. All the compositions resulting in $cd_i, acd_i, cdb_i$ are included in the appropriate labels due to her strategy, see Figure~\ref{fig:cd}. Otherwise the move can be satisfied with a new node, satisfying all the properties in $\exists$'s strategy. Similarly, the argument can be constructed for range moves including $d_i$ or otherwise.

In case a witness move is called and the left operand is $c_i$ or the right operand is $d_i$ (or both), the witness node returned must be the designated $y$ and the induction hypothesis is maintained (again, see Figure~\ref{fig:cd}). If the witness move has $cd_i$ as an operand, she makes sure to designate the appropriate $y$, again preserving the induction hypothesis. All other non-redundant operations result in $ab_i$. If the index of the operands is $i-1$, $\exists$ may ensure she does not include $a_{i-1}*b_{i-1}$ witness (see Figure~\ref{fig:wit} left). Finally, for operands with index $i$ she adds a witness node with $m_i, cd_i$ in the reflexive label, $a_i, acd_i$ on the left and $b_i, cdb_i$ on the right (see Figure~\ref{fig:wit} right). This covers all the possible non-redundant witness moves, but results in $ab_{i+1}$ being added to the label. In any case, the induction hypothesis is maintained.

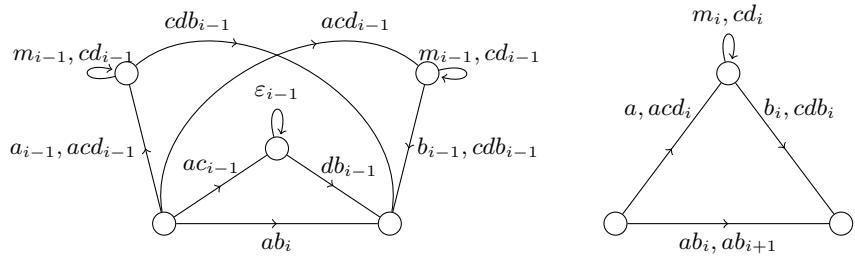
\begin{figure}
    \centering
    \begin{tikzpicture}
        \node[draw, circle] (x) at (0.5,0){};
        \node[draw, circle] (y) at (3.5,0){};
        \node[draw, circle] (z1) at (0,2){};
        \node[draw, circle] (z2) at (2,1){};
        \node[draw, circle] (z3) at (4,2){};
        
        \path (x) edge[->-] node[below]{$ab_i$} (y);
        \path (x) edge[->-] node[left]{$a_{i-1}, acd_{i-1}$} (z1);
        \path (z1) edge[->--, bend left=70] node[above, pos=0.2]{$cdb_{i-1}$} (y);
        \path (x) edge[-->-, bend left=70] node[above, pos=0.8]{$acd_{i-1}$} (z3);
        \path (z3) edge[->-] node[right]{$b_{i-1}, cdb_{i-1}$} (y);
        \path (x) edge[->-] node[left, pos=0.8]{$ac_{i-1}$} (z2);
        \path (z2) edge[->-] node[right, pos=0.2]{$\;db_{i-1}$} (y);
        
        \path (z2) edge[loop above] node[above]{$\varepsilon_{i-1}$} (z2);
        \path (z1) edge[loop left] node[above]{$m_{i-1}, cd_{i-1}\;\;\;\;$} (z1);
        \path (z3) edge[loop right] node[above]{$\;\;\;\;m_{i-1}, cd_{i-1}$} (z3);
        
        \node[draw, circle] (x) at (6.5,0){};
        \node[draw, circle] (y) at (9.5,0){};
        \node[draw, circle] (z2) at (8,2){};
        
        \path (x) edge[->-] node[below]{$ab_i, ab_{i+1}$} (y);
        \path (x) edge[->-] node[left, pos=0.8]{$a, acd_{i}$} (z2);
        \path (z2) edge[->-] node[right, pos=0.2]{$b_i, cdb_{i}$} (y);
        \path (z2) edge[loop above] node[above]{$m_{i}, cd_i$} (z2);
    \end{tikzpicture}
    \caption{Witness moves for $ab_i$ with $i-1$ left and $i$ on the right}
    \label{fig:wit}
\end{figure}

We have now seen that $\exists$ can play a game in a way that the only possible inconsistency that can arise is from $ab_i \in \top(x,y)$ and also in $\bot(x,y)$. Without loss, this situation can only arise when the initialisation pair is $ab_0, ab_{n+1}$. In this case she plays the initial non-reflexive network with $ab_0 \in \top(x,y), ab_{n+1} \in \bot(x,y)$. As she can only increase the maximal $i$ such that $ab_i \in \top(x,y)$ by 1 each move, she introduces an inconsistency at the $n+1$st move at the earliest. Thus she can win $\Gamma_n(\s_n)$.
\qed\end{proof}
This gives us all we need to conclude

\begin{theorem}
$R(\Dom, \Ran, *)$ cannot be axiomatised by a finite first order theory.
\end{theorem}

\begin{proof}
Suppose such a theory existed, call it $\Psi$. Then $R(\Dom, \Ran, *)$ is axiomatised by a single axiom $\psi = \bigwedge_{\psi' \in \Psi} \psi'$. Thus $\Sigma \cup \{\neg \psi\}$ is not consistent as, by Lemma~\ref{lem:sigma}, $\Sigma$ ensures that any model of it is representable and $\neg \psi$ ensures it is not. Now look at any finite subtheory $\Omega \subseteq \Sigma \cup \{\neg \psi\}$. Observe how, since it is finite, there exists $n < \omega$ such that for all $m>n$ we have $\sigma_m \not\in \Omega$. Thus $\s_n \models \Omega$ as by Lemmas~\ref{lem:ws},~\ref{lem:sigma} we have $\s_n \models \sigma_i, i\leq n$, and by Lemmas~\ref{lem:winStrat},~\ref{lem:nonrep} we have $\s_n \models \neg\psi$. By compactness of first order logic, we conclude the Theory $\Sigma \cup \{\neg \psi\}$ is consistent and we have reached a contradiction.
\qed\end{proof}

\section{Finite Representation Property}
\label{sec:frp}
We have now seen that both the angelic and demonic representable domain-range semigroups cannot be axiomatised finitely. However, it remains unknown if all finite members of $R(\Dom, \Ran, ;)$ and $R(\Dom, \Ran, *)$ have the finite representation property. Although the finite axiomatisability (or lack thereof) is known for a number of representation classes \cite{mikulas2004axiomatizability}, FRP remains largely unknown for signatures with composition. In this section we discuss some existing results and extend FRP result for ordered domain algebras \cite{hirsch2013meet}.

The known results regarding FRP are summarised in Table~\ref{tab:frp}. The signatures $\{;\}, \{1',;\}, \{\Dom, *\}$ are well known examples where Cayley representation for groups can be used to represent the structure over a finite base. Neuzerling shows that any signature containing meet and composition fails to have FRP using Point Algebra \cite{neuzerling2016undecidability}. In \cite{hirsch2020finite} we show that this structure can also be used to show that FRP fails for any signature containing negation, partial order and composition. In a forthcoming paper, we extend this result to any signature containing $\{-, ;\}$.

\begin{table}
    \caption{Signatures with composition where FRP is known}
    \centering
    \begin{tabular}{|p{5cm}|p{5cm}|}
        \hline
        \textbf{FRP} & \textbf{No FRP} \\\hline
        $\{;\}, \{1', ;\}, \{\Dom,*\}$ & $\{\cdot, ;\} \subseteq \tau$ \cite{neuzerling2016undecidability} \\
        $\{\leq, ;\}, \{\sqsubseteq, *\}$ \cite{zareckiui1959representation} & $\{-, ;\} \subseteq \tau$ \\
        $\{0,1,\Dom, \Ran, \leq, 1', \smile, ;\}$ \cite{hirsch2013meet} &\\
        $\{\sqsubseteq, ;\}$ \cite{hirsch2020finite} & \\
        $\{\leq, \setminus, / , ;\}$ \cite{rogozin2020finite} & \\\hline
    \end{tabular}
    \label{tab:frp}
\end{table}

A simple approach to constructing a finite representation of a relational partially ordered semigroup was proposed by Zarecki{\u \i} in \cite{zareckiui1959representation} where one may amend a representable $\{\leq, ;\}$-structure $\s$ with a compositional identity element $e$ and only add the mandatory $(e,e)$ to $\leq$ to then define a simple representation $\theta$ over the base $\s$ with
$$(s,t) \in a^\theta \Longleftrightarrow t \leq s;a$$
The inclusion of $e$ ensures faithfulness as for $a \not \leq b$ $(e,a) \in a^\theta \setminus b^\theta$ and the associativity and monotonicity ensure that $\leq, ;$ are correctly represented.

Egrot and Hirsch \cite{hirsch2013meet} amend the idea to represent the ordered domain algebras,  the signature $\{0, \Dom, \Ran, \leq, 1', \smile, ;\}$ where $0$ is the empty relation (bottom element of the Boolean lattice), $1'$ is the relational identity and $\smile$ is the relational converse. They represent the structures in $R(0, \Dom, \Ran, \leq, 1', \smile, ;)$ over the base of subsets of the structure, rather than its elements.

However, their result can be adapted for a wider range of signatures. Below we present an outline of the proof for the following theorem.

\begin{proposition}
\label{prop:simclass}
For any signature $\{\Dom, \Ran, \smile, ;\} \subseteq \tau \subseteq \{ 0, 1, \Dom, \Ran,\leq, 1', \smile, ;\}$, $R(\tau)$ has the finite representation property.
\end{proposition}

\begin{proof}
We can, for any representable $\tau$-structure $\s$, define a partial ordering $\leq$ (even if $\leq\; \notin \tau$) as the set of all pairs where $s \leq t$ if and only if for all representations $\theta$, $s^\theta \leq t^\theta$. Similarly, one can define at most one element $0$ (again even if $0 \notin \tau$) that will always be represented as an empty relation.

This means that we can define the set of closed sets $\G$ as the set of all $\emptyset \subsetneq S \subseteq \s \setminus \{0\}$ such that for $\Dom(S) = \prod_{s \in S} \Dom(s)$ and similarly $\Ran(S)$, we have $(\Dom(S);S;\Ran(S))^\uparrow = S$ where $\uparrow$ is upward closure with respect to $\leq$. Then define a mapping $\rho: \s \rightarrow \wp(\G \times \G)$ such that $(S,T) \in a^\rho$ if and only if $S;a \subseteq T$ and $T;\Breve{a} \subseteq S$.

The mapping is faithful as for $a \nleq b$, $(\Dom(a), a) \in a^\rho$ as $a;\Breve{a} \geq \Dom(a)$, but not in $b$ as that would mean $a \leq \Dom(a);b \leq b$. It represents $\leq$ correctly by monotonicity of $;$ over $\leq$ and $0, 1$ correctly as $1$ is the top element with respect to ordering and $a;0 = 0;a = 0$, for all $a$. Domain and range are correctly represented as if there is an outgoing/incoming edge from $S$ with $a$/$\Breve{a}$, then $S;a;\Breve{a} \subseteq S$ and since $\Ran(\Breve{a}) = \Ran(a;\Breve{a}) = \Dom(a)$, $S;\Dom(a) \subseteq S$ and thus $\Dom(a) = \Ran(\Breve{a})$ is included in $(S,S)$. Furthermore if $\Ran(a) = \Dom(\Breve{a})$ is included in $(S,S)$ then $(S;\Breve{a})^\uparrow$ ensures that there is an incoming edge with $a$ and an outgoing edge with $\Breve{a}$. Finally, domain elements are only on reflexive nodes as if $(S;\Dom(a))^\uparrow = S$ so if $(S,T) \in \Dom(a)$ then $S \subseteq T \subseteq S$ and similarly $(S,T) \in (1')^\rho$ if and only $S = T$. Converse is correctly represented as $\Breve{\breve{a}} = a$. Finally $a^\rho;b^\rho \leq (a;b)^\rho$ by monotonicity and $(a;b)^\smile = \breve{a};\breve{b}$ and $(a;b)^\rho \leq a^\rho;b^\rho$ as if $(S,T) \in (a;b)^\rho$, $\Big(S;a;\Dom(\breve{a};\breve{T}) \cup T;a;\Ran(S;a)\Big)^\uparrow$ is an appropriate witness for the composition.
\qed\end{proof} 
Note that the second part of the proof where we show that $\rho$ is indeed a representation is an outline. This is because the argument closely follows that in \cite[Section 6]{hirsch2013meet}, refer to it for more detail.

Finally, Rogozin shows that one can embed residuated semigroups into relational quantales in \cite{rogozin2020finite} and we show in \cite{hirsch2020finite} that a Zarecki{\u \i} representation can be modified in a way to represent semigroups with demonic refinement. The latter was the first example of a signature with composition without a finitely axiomatisable representation class, but with FRP.

\section{Problems}
In this section we look at some open problems and outline the difficulties with showing the finite representation property.

We begin with the observation that $e$ in the Zarecki{\u \i} representation, as defined in Section~\ref{sec:frp}, is not represented as the true relational identity element, i.e. $1' = \{(x,x) \mid x \in X\}$, as for some $a \lneq a'$ we will have $(a', a) \in e^\theta$. Thus this good behaviour does not extend to the signature of $\{1', \leq, ;\}$, with $R(\leq, 1', ;)$ non-finitely axiomatisable \cite{hirsch2011axiomatizability} and FRP unknown.

$R(\leq, 1', ;)$ suffers from the same problem as $R(\Dom, \Ran, *)$ and $R(\Dom, \Ran, ;)$. That is, some elements are always represented as \emph{partial functions}, that is, for any representation $\theta$ over $X$, if $(x,y) \in f^\theta, (x,z) \in f^\theta$ then $y = z$. Simple examples of that include the domain-range elements, as well as those $f \leq 1'$. However, composition makes for some more interesting examples, like $c_i$ in $\s_n$ in Section~\ref{sec:nfa} or in $R(\Dom, \Ran, ;)$, $\Ran(a);b$ will always be represented as a partial function if $\Dom(a;b) = a;b$. This is illustrated in Figure~\ref{fig:pf}, from left to right, observe how for any representation $\theta$ if $(x,y) \in (R(a);b)^\theta$ then $(x,x) \in R(a)^\theta$, so there must exist a $z$ such that $(z,x) \in a^\theta$. As $a;b = \Dom(a;b)$ and by composition, $z$ must be the same as $y$. Similarly, for any outgoing $z$ with $(x,z) \in (R(a);b)^\theta$, it has to be the case that $y=z$. 

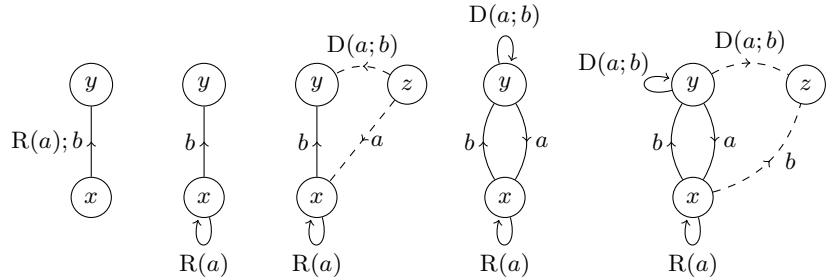
\begin{figure}
    \centering
    \begin{tikzpicture}
        \node[circle, draw](x1) at (0,0) {$x$};  
        \node[circle, draw](y1) at (0,1.5){$y$};
        \node[circle, draw](x2) at (1.5,0) {$x$};  
        \node[circle, draw](y2) at (1.5,1.5){$y$};
        \node[circle, draw](x3) at (3,0) {$x$};  
        \node[circle, draw](y3) at (3,1.5){$y$};
        \node[circle, draw](z3) at (4.2,1.5){$z$};
        \node[circle, draw](x4) at (5.5,0) {$x$};  
        \node[circle, draw](y4) at (5.5,1.5){$y$};
        \node[circle, draw](z5) at (9.5,1.5){$z$};
        \node[circle, draw](x5) at (8,0) {$x$};  
        \node[circle, draw](y5) at (8,1.5){$y$};
        
        \path (x1) edge[->-] node[left]{$\Ran(a);b$} (y1);
        
        \path (x2) edge[->-] node[left]{$b$} (y2);
        \path (x2) edge[loop below] node[below]{$\Ran(a)$} (x2);
        
        \path (x3) edge[->-] node[left]{$b$} (y3);
        \path (x3) edge[loop below] node[below]{$\Ran(a)$} (x3);
        \path (z3) edge[->-, dashed] node[right]{$a$} (x3);
        \path (z3) edge[->-, dashed, bend right] node[above]{$\Dom(a;b)$} (y3);
        
        \path (x4) edge[loop below] node[below]{$\Ran(a)$} (x4);
        \path (x4) edge[->-, bend left] node[left]{$b$} (y4);
        \path (y4) edge[->-, bend left] node[right]{$a$} (x4);
        \path (y4) edge[loop above] node[above]{$\Dom(a;b)$} (y4);
        
        \path (x5) edge[loop below] node[below]{$\Ran(a)$} (x5);
        \path (x5) edge[->-, bend left] node[left]{$b$} (y5);
        \path (y5) edge[->-, bend left] node[right]{$a$} (x5);
        \path (y5) edge[loop left] node[above]{$\Dom(a;b)\;\;\;\;\;\;\;\;\;$} (y5);
        \path (y5) edge[->-, dashed, bend left] node[above]{$\Dom(a;b)$}(z5);
        \path (x5) edge[->-, dashed, bend right] node[right]{$\;b$}(z5);
    \end{tikzpicture}
    \caption{Partial-Functional Nature of $R(a);b$ when $a;b = \Dom(a;b)$}
    \label{fig:pf}
\end{figure}

Every function in the signature of domain-range algebras comes with a converse. More specifically, if $D(a;b) = a;b$ then not only is $R(a);b$ a function, but $a;D(b)$ is its well defined converse. Unfortunately, this does not enable us to use represent structures over a finite base in the same way as the structures in Proposition~\ref{prop:simclass}.

It is true that partial functions, their converses and arbitrary compositions of those have their converse well defined. But take an $a$ with its converse defined and say $a = b;c$ and $\Ran(b) = \Dom(c)$. Observe that converses of $b,c$ not defined. Both $b$ and $c$ have a \emph{partial converse}. That is, for every representation $\theta$, $(b^\theta)^\smile \leq (c;\breve{a})^\theta$ and $(c^\theta)^\smile \leq (\breve{a};b)^\theta$, but the $\geq$ inclusions do not necessarily hold, see Figure~\ref{fig:pconv}. 

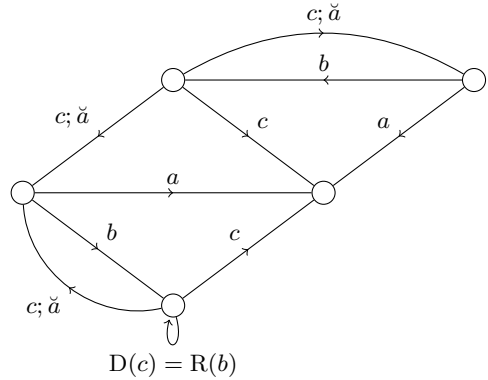
\begin{figure}[ht]
    \centering
    \begin{tikzpicture}
        \coordinate[draw, circle] (x) at(0,0);
        \coordinate[draw, circle] (y) at(4,0);
        \coordinate[draw, circle] (z) at(2,1.5);
        \coordinate[draw, circle] (z1) at(2,-1.5);
        \coordinate[draw, circle] (w) at(6,1.5);
        
        \path(x) edge[->-] node[above right]{$b$}(z1);
        \path(z1) edge[->-] node[above left]{$c$}(y);
        \path(z1) edge[loop below] node[below]{$\Dom(c) = \Ran(b)$} (z1);
        \path (z1) edge[->-, bend left=50] node[below left]{$c;\breve{a}$} (x);

        \path(x) edge[->-] node[above]{$a$}(y);
        \path(w) edge[->-] node[above]{$b$}(z);
        \path(z) edge[->-, bend left] node[above]{$c;\breve{a}$}(w);
        \path(z) edge[->-] node[above right]{$c$}(y);
        \path (z) edge[->-] node[above left]{$c;\breve{a}$} (x);
        \path (w) edge[->-] node[above left]{$a$} (y);
    \end{tikzpicture}
    \caption{Partial converse of $b$, i.e. $b \leq c;\breve{a}$, but $\breve{c};a \not \leq b$, where $a,b,c$ are elements of a domain range semigroup}
    \label{fig:pconv}
\end{figure}

This enables us to define the partial converse of $s \in \s$ to be the set $\C(s) \subseteq \s$ where $\C(s)$ is the set of all $s' \in \s$ such that $(s^\theta)^\smile \leq (s')^\theta$, for any representation $\theta$. However, as we have seen there is no guarantee that $\C(\C(s))=s^\uparrow$. Furthermore, $\C(t);\C(s) \subseteq \C(s;t)$ but not necessarily $\C(t);\C(s) \supseteq \C(s;t)$. As the proof of FRP for ordered domain algebras heavily relies on both $\breve{\breve{a}} = a$ and $(a;b)^\smile = \breve{b};\breve{a}$, the same representation cannot be used for converse-free signatures.

Adding join ($+$) to the signature adds additional difficulty. The class of representable join-lattice semigroups $R(+, ;)$ was shown non-finitely axiomatisable in \cite{andreka1988representation}, with the finite representation property remaining open. Similar to the case where $1'$ is added to the signature of $\{\leq, ;\}$, this slight modification completely breaks the Zarecki{\u \i} representation. That is because $+$ is not necessarily \emph{distributive}, i.e. if $a \leq b+c$ there exists some $b' \leq b$ and $c' \leq c$ such that $a = b' + c'$.

For distributive lattices, one can define the Zarecki{\u \i} representation over the set of minimal non-$0$ elements and preserve all operations in a faithful manner. However, no signature including $\{+, ;\}$ has been shown to have the finite representation property for its representation class thus far.

The problems raised in this section can be summarised below

\begin{problem}
Do converse-free (ordered) domain-range semigroups have the finite representation property? How about their demonic counterparts?
\end{problem}

\begin{problem}
Do signatures containing the join-semilattice and composition have the finite representation property?
\end{problem}

\begin{problem}
Does $R(\leq, 1', ;)$ have FRP? How about $R(\leq, 1', \smile, ;)$ or $R(\leq, \smile, ;)$?
\end{problem}
%
%
%
%
%
\bibliographystyle{splncs04}
\bibliography{ref}
\end{document}